\documentclass[aps,prl,reprint,groupedaddress]{revtex4-2}

\usepackage{amsfonts, amsmath, amsthm, comment, braket}
\usepackage{tikz}
\usetikzlibrary{arrows.meta,calc}

\usepackage[linktocpage, colorlinks=true, linkcolor={blue}, citecolor={blue}]{hyperref}
\newtheorem{theorem}{Theorem}
\newtheorem{example}{Example}
\newtheorem{corollary}[theorem]{Corollary}

\newtheorem{proposition}[theorem]{Proposition}
\newtheorem{definition}[theorem]{Definition}
\newtheorem*{question}{Question}

\newcommand{\mb}{\mathbb}
\newcommand{\mc}{\mathcal}

\newcommand{\lset}{\left\{ }
\newcommand{\rset}{\right\} }

\newcommand{\lpara}{\left(}
\newcommand{\rpara}{\right)}

\newcommand{\norm}[1]{\left\Vert #1 \right\Vert}

\DeclareMathOperator{\id}{Id}
\DeclareMathOperator{\tr}{tr}

\DeclareMathOperator{\diag}{diag}
\DeclareMathOperator{\rankshift}{rankshift}
\DeclareMathOperator{\rank}{rank}

\begin{document}

\title{The Hidden Nature of Non-Markovianity}

\author{Jihong Cai}
\email[]{jihongc2@illinois.edu}
\affiliation{Department of Mathematics, University of Illinois Urbana-Champaign}

\author{Advith Govindarajan}
\email[]{agovind2@illinois.edu}
\affiliation{Department of Mathematics, University of Illinois Urbana-Champaign}

\author{Marius Junge}
\email[]{mjunge@illinois.edu}
\affiliation{Department of Mathematics, University of Illinois Urbana-Champaign}

\date{\today}

\begin{abstract} 
    The theory of open quantum systems served as a tool to prepare entanglement at the beginning stage of quantum technology and more recently provides an important tool for state preparation. Dynamics given by time dependent Lindbladians are Markovian and lead to decoherence, decay of correlation and convergence to equilibrium. In contrast Non-Markovian evolutions can outperform their Markovian counterparts by enhancing memory. In this letter we compare the trajectories of Markovian and Non-Markovian evolutions starting from a fixed initial value. It turns out that under mild assumptions every trajectory can be obtained from a family of time dependent Lindbladians. Hence Non-Markovianity is invisible if single trajectories are concerned.
\end{abstract}
\maketitle

\paragraph*{Introduction.}
The theory of open quantum systems describes the evolution of quantum devices including interactions with an environment. Open quantum systems have been studied extensively in quantum information \cite{breuer2002, nielsen2010}, control \cite{wiseman2009, dong2010, altafini2012}, and thermodynamics \cite{spohn1978, lindblad2001, vinjanampathy2016}. For completely positive and trace-preserving (CPTP) Markovian dynamics, Gorini, Kossakowski, Sudarshan \cite{gorini1976}, and Lindblad \cite{lindblad1976} identified the general form of the time-local generator of the GKSL master equation
\begin{align}\label{eom}
    \dot\rho_t=L_t\rho_t,
\end{align}
where $L_t$ is a Lindbladian operator. Independently, the same form is derived from the Born-Markov approximation \cite{rivas2012, breuer2002}. Data processing inequality, decoherenence, and decay of correlation are imminent to the Markovian regime. 

Many experimentally relevant processes exhibit memory effects, such as information backflow \cite{liu2011}, non-monotonic decoherence \cite{tang2012}, or revival of entanglement \cite{xu2010}. These observations have motivated sustained efforts to characterize quantum dynamics beyond the Markovian regime. Such effect can be demonstrated in many quantum devices, such as nitrogen-vacancy centers \cite{dong2018, haase2018, wang2018}, photonic systems \cite{xu2010, liu2011, liu2013}, nuclear magnetic resonance \cite{wu2020, khurana2019, chen2022}, trapped ions \cite{gessner2014, li2022}, and on superconducting processors \cite{white2020, gaikwad2024}.

Multiple definitions of quantum Non-Markovianity have been proposed, reflecting different structural and operational perspectives. Foundational work by Wolf and Cirac \cite{wolfcirac2008} and Wolf, Eisert, Cubitt, and Cirac \cite{wolf2008} initiated a systematic analysis of Markovianity at the level of quantum dynamical maps. Despite significant progress, no single definition has achieved universal acceptance (see Refs. \cite{breuer2016, rivas2014} for reviews): CP-divisibility \cite{rivas2010}, information backflow \cite{breuer2009}, mutual information \cite{luo2012}, channel capacities \cite{bylicka2014}, decay-rate negativity \cite{chruscinski2014}, trajectory-based formulations \cite{piilo2008}, tensor network \cite{jorgensen2019, link2024}, and process tensors \cite{pollock2018}. Ongoing work continues to clarify relations between different notions \cite{li2018}. While differing in formulation, these approaches consistently regard Markovianity as a property of the reduced dynamical map or, more generally, of the associated multi-time quantum process.

From a naïve standpoint, only trajectories can be experimentally observed. More precisely, given a CPTP evolution $T_t$, observed trajectories are given by $\rho_t = T_t \rho_0$ for an initial state $\rho_0$. Such a trajectory $(\rho_t)_{t_0 \leq t  < t_1}$ is a prototype of a path of density matrices (mixed states). From individual realizations of dynamics in the form of trajectories ensemble descriptions can be inferred. This has motivated extensive use of trajectory-based methods, including stochastic unravelings and continuous-measurement descriptions \cite{carmichael1993, dalibard1992, wiseman1993, wiseman2009, korotkov1999, katz2008, murch2013}. This leads to the following question: 
\begin{center}
    \textit{Can Non-Markovianity be detected from the trajectory?}
\end{center}

In this Letter we show that under mild assumptions, the answer is a resounding \textit{No}. We prove that Non-Markovianity is not an invariant property of individual quantum trajectories, not even finite ensembles thereof. This is demonstrated on a trajectory induced by the pure dephasing evolution, such as 
$$\rho_t=\frac12\begin{pmatrix}
    1 & \sin(t)\\
    \sin(t) & 1
\end{pmatrix}.$$
More general classes of such example will be discussed later in the Letter.

These results do not alter existing definitions of quantum Non-Markovianity formulated at the level of dynamical maps or multi-time processes. Rather, they establish the fundamental non-identifiability of Markovianity under trajectory-preserving descriptions.

\paragraph*{Preliminaries.}
Let $\mc H$ be a finite-dimensional Hilbert space and denote $\mc D(\mc H)=\lset \rho\in\mb B(\mc H): \rho\geq 0, \tr\rho=1\rset$ the set of density matrices in $\mc H$. A trajectory in
$\mc D(\mc H)$ is a continuous map
$$[t_0,t_1) \to \mc D(\mc H),\quad t \mapsto \rho_t.$$

We call an evolution $T_t$ Markovian if it is given by time-dependent Lindbladians. More precisely, it satisfies a
time-local master equation of the Gorini-Kossakowski-Sudarshan-Lindblad
(GKSL) form Eq. (\ref{eom})
where the generator $L_t$ is a Lindbladian. As shown in the seminal works
of Gorini \emph{et al.} and Lindblad \cite{gorini1976,lindblad1976}, the
most general form of a time-local generator compatible with CPTP maps is
\begin{equation}\label{lin}
\begin{aligned}
    L_t(\rho) &= -i[H_t,\rho] + \sum_i \gamma_i(t) L_{a_i(t)} (\rho),\\
    L_a(\rho) &= a\rho a^\dagger - \tfrac12\{a^\dagger a,\rho\}
\end{aligned}
\end{equation}
where $H_t$ is a (possibly time-dependent) Hamiltonian,
$a_i(t)$ are jump operators, and the rates $\gamma_i(t)\ge 0$. Time-local Lindbladians generate Markovian evolutions $T_t$ via the time-ordered exponential $T_t = \mc T \lset \exp \left(\int_0^t L_s\, ds\right) \rset$ (equivalently solving the differential equation).

\paragraph*{Main Results.}
We show that trajectories generated by Non-Markovian dynamics can always be reproduced by suitably chosen Markovian evolutions. Under mild assumptions, we prove that \textit{any} differentiable path of density matrices $\rho_t$ can be realized as the trajectory of a Markovian dynamical map $T_t$. We apply this result to several classes of trajectories naturally arising from Non-Markovian evolutions, demonstrating that Markovianity cannot be inferred from a single trajectory. We further show that even collections of trajectories may remain indistinguishable in the absence of additional selection criteria.

We impose two mild regularity assumptions on the path of density operators, ensuring well-behaved spectral properties:
\begin{enumerate}
\item[(1)] The eigenspaces of $\rho_t$ depend continuously on $t$
\item[(2)] The set of times at which $\rho_t$ changes rank has no accumulation point.
\end{enumerate}
These conditions exclude pathological spectral behavior and are satisfied by a wide range of Markovian and Non-Markovian evolutions of physical interest.

For a given trajectory $\rho_t$ we call $L_t$ a \emph{Lindbladian lift} for $\rho_t$ if it satisfies the equation of motion
\begin{equation}
    \dot \rho_t = L_t \rho_t. \tag{LL} \label{LL}
\end{equation}
Constructing the corresponding channels amounts to a memory destroying explanation of the path $\rho_t$ as a trajectory of a Markovian evolution. Under these assumptions, we obtain the following result:
\begin{theorem}[Lindbladian Lifting] \label{lifting}
    For any $C^2$ path $\rho_t$ satisfying properties (1) and (2), there exists a continuous Lindbladian lift on the interval $[t_0,t_1)$.
\end{theorem}
We refer readers to the supplementary material \cite{supp_material} for the proof of the theorem and the following corollaries.

Since conditions (1) and (2) may be technically involved to verify, it is useful to state two simpler criteria.
\begin{corollary}
    If $\rho_t$ is a $C^k$ path  has no non-zero repeated eigenvalues on $t_0 \leq t < t_1$, then there is a $C^k$ family of time-local Lindbladians for which $\dot \rho_t = L_t \rho_t$ on $t_0 \leq t < t_1$. 
\end{corollary}
The following corollary provides a constructive criterion.
\begin{corollary}
    Take a $C^k$ path of densities $\rho_t$. If for all points $s \in [t_0,t_1)$ where $\det \rho_s = 0 $ we have that  $\dot \rho_s = 0$ then $(\rho_t)_{t_0 \leq t < t_1}$ has a Lindbladian lift $L_t$ where $L_t$ is $C^k$.
\end{corollary}
In the setting of Corollary 3, the generator $L_t$ can be given explicitly as
$$L_t=\frac{1}{\varepsilon_t}\lpara \mc R_{\rho_t+\varepsilon_t\dot\rho_t} -\id\rpara, \quad \varepsilon_t=e^{1-\frac{1}{(1-t)^2}}.$$ where $$\mc R_\sigma(\rho)=\tr(\rho)\sigma$$ is the replacer channel to the density $\sigma$. We call such $L_t$ \emph{replacer Lindbladians}.

\paragraph*{Single-path Indistinguishability.}
Consider the path of densities given as
$$\rho_t=\frac12\begin{pmatrix}
    1 & \sin(t)\\
    \sin(t) &1
\end{pmatrix}$$
This is canonically given by the pure dephasing evolution with time-local generator 
$$N_t(\rho)=\gamma(t)\lpara Z\rho Z-\rho\rpara, \quad \gamma(t)=\cot(t),$$
where $Z$ is the Pauli-$Z$ matrix, applied to the initial state 
$$\rho_0=\frac12\begin{pmatrix}
    1 & 1\\
    1 & 1
\end{pmatrix}.$$
The rate $\gamma(t)$ changes sign, and diverges at $t=\pi/2$, indicating Non-Markovian dynamics.
This is the propotype of many NM evolutions \cite{addis2014, bylicka2014, haikka2013}

We now construct a Markovian realization of the same trajectory. Owing to periodicity, it suffices to consider $t\in[0,2\pi]$. Note that $\det \rho_t =0$ only at the point $t=\pi/2$, where
$$\dot\rho_{\pi/2}=0.$$
This allows for a smooth Lindbladian lift of the form
$$L_t=\frac{1}{\varepsilon_t}\lpara \mc R_{\rho_t+\varepsilon_t\dot\rho_t}-\id\rpara, \quad \varepsilon_t=e^{1-\frac{1}{(1-t)^2}}.$$
where 
$$\dot\rho_t=\frac12 \begin{pmatrix}
0 & \cos(t)\\
\cos(t) & 0
\end{pmatrix}$$
Since $\rho_t$ is smooth, the
resulting family $\{L_t\}_t$ defines a smooth Markovian evolution realizing the trajectory.

We can see that for a single path of densities $\rho_t$ one cannot determine if it is a result of a Non-Markovian or Markovian evolution without additional information. In fact, by Theorem \ref{lifting}, almost all paths can be realized as a time-dependent Lindbladian. This indistinguishability result holds for collections of paths.

\paragraph*{Exponential Indistinguishability.} Given an evolution, how many paths do we have to check to see whether it is Non-Markovian or Markovian? The previous example shows that a single path is not enough. Here we construct a case where it is not enough to check $2^n$ paths where $n$ is the number of qubits.

We start with a path $$\rho_t = \frac14\begin{pmatrix}
    2 & 1+e^{-2t}\\
    1+e^{-2t} & 2
\end{pmatrix}.$$ This path can be realized as a trajectory of the eternally Non-Markovian qubit evolution introduced by Hall, Cresser, Li, and Andersson \cite{hall2014}, given by
\begin{align*}
    &T_t
\begin{pmatrix}
    \rho_{11} & \rho_{12}\\
    \rho_{12}^* & 1-\rho_{11}
\end{pmatrix}\\
&\qquad=
\begin{pmatrix}
    \tfrac12(1-e^{-2t})+\rho_{11}e^{-2t}
    &
    \tfrac12(1+e^{-2t})\rho_{12}
    \\
    \tfrac12(1+e^{-2t})\rho_{12}^*
    &
    \tfrac12(1+e^{-2t})-\rho_{11}e^{-2t}
\end{pmatrix}.
\end{align*}
This evolution is Non-Markovian for all $t>0$ due to persistent
negativity of the decay rates.

The state $\rho_t$ has full rank for all $t>0$ which yields a Markovian lift of the form $$L_t=\varepsilon_t^{-1}\lpara \mathcal R_{\rho_t+\varepsilon_t\dot\rho_t} -\id\rpara, \quad \dot\rho_t = \frac12 \begin{pmatrix}
0 & -e^{-2t}\\
-e^{-2t} & 0
\end{pmatrix}$$
where $\varepsilon_t$ is the same as before.

Note that $\rho_t$ has the asymptotic state
$$\rho_\infty= \frac14\begin{pmatrix}
2 & 1\\
1 & 2
\end{pmatrix}.$$ Now for any subset $S\subseteq\{1,\dots,n\}$, define the product paths
$$\rho^{(S)}_t=\bigotimes_{i=1}^n \rho_{i,t}^{(S)},\quad \rho_{i,t}^{(S)}= \begin{cases}
    \rho_t, & i\in S,\\
    \rho_\infty, & i\notin S.
\end{cases}$$
Then each of the $2^n$ paths $\rho^{(S)}_t$ is generated by the
same Lindbladian lift $\bigotimes_{i=1}^n L_t^i$, where $L_t^i$ is $L_t$ on the $i$-th register and identity everywhere else. Simultaneously, each of the $2^n$ paths is generated by the same Non-Markovian evolution $\bigotimes_{i=1}^n T_t^i$. Hence, an exponentially
large family of (supposedly Non-Markovian) trajectories can
be simultaneously realized by a single Markovian dynamics. Moreover, this evolution is of particular interest because it also cannot be determined to be Non-Markovian by typical distance measures \cite{hall2014}.

\paragraph*{Regional Indistinguishability for negative jump Lindbladian.}
Consider another dephasing channel
$$T_t(\rho)= \begin{pmatrix}
    \rho_{11} & \rho_{12} & \rho_{13}\sin(t)\\
    \rho_{21} & \rho_{22} & \rho_{23}\sin(t)\\
    \rho_{31}\sin(t) & \rho_{32}\sin(t) & \rho_{33}
\end{pmatrix},$$
which is Non-Markovian due to the oscillatory factor $\cot(t)$.
The time-local generator is given by 
$N_t=\cot(t) L_{\ket 3\bra 3}$
with jump operator $\ket3\bra3$.
Fix complex parameters $a,b$ with $|a|^2 + |b|^2 \leq \frac19$ and consider the trajectory
$$\rho_t= \begin{pmatrix}
    \frac13 & 0 & a\sin(t)\\
    0 & \frac13 & b\sin(t)\\
    a^*\sin(t) & b^*\sin(t) & \frac13
\end{pmatrix}.$$
Any time-local Lindbladian $L_t$ satisfying $L_t(\rho_t)=\dot\rho_t$
defines a lifting of this path.

Crucially, such a lifting can be extended by restriction to any trajectory
of the form
$$\tilde\rho_t= \begin{pmatrix}
    \frac13 & z & a\sin(t)\\
    z^* & \frac13 & b\sin(t)\\
    a^*\sin(t) & b^*\sin(t) & \frac13
\end{pmatrix},$$
with arbitrary $z$ close enough to $0$. Varying the parameter $z$ results in a connected region of paths. Consequently, a single Markovian generator $\tilde L_t$ simultaneously lifts an entire
connected region of trajectories arising from the same Non-Markovian evolution.  We refer to supplementary material \cite{supp_material} to see how the construction $L_t$ arises from the fixed parameters $a,b$. 

\paragraph*{Lifting Discrete Points.}
Since many characterizations of quantum dynamics are based on finite time samples \cite{tang2012, liu2011, cialdi2019, bernardes2016}, it is natural to ask whether such information suffices to detect Non-Markovianity.

Let $\sigma_1,\dots,\sigma_N \in \mc D(\mc H)$ be density operators specified at times
$0=t_1<\dots<t_N$.
Define a piecewise affine interpolation $\rho_t$ by
$$\rho_t = \frac{t_{k+1}-t}{t_{k+1}-t_k}\sigma_k + \frac{t-t_k}{t_{k+1}-t_k}\sigma_{k+1}, \quad t\in[t_k,t_{k+1}].$$
Then $\rho_{t_k}=\sigma_k$ for all $l$, and $\rho_t$ a piecewise affine path.

By Proposition 6 in the supplementary material \cite{supp_material}, any such piecewise affine path admits a realization by a time-local Lindbladian, yielding a Markovian evolution whose trajectory coincides with $\rho_t$.
Consequently, any finite set of discrete time samples arising from an arbitrary evolution can be reproduced by a Markovian dynamics, independently of whether the original evolution was Markovian or Non-Markovian.

Thus, discrete trajectory data do not suffice to detect Non-Markovianity, which remains a genuinely global property of the dynamical map.

\paragraph*{Discussion.}
On a single path, Non-Markovianity is not visible. While there are well-known obstructions for specially chosen pairs of trajectories, such as violations of the data processing inequality, information backflow, non-monotonic decay rates, or monotone channel functionals \cite{gour2021, saxena2020}, their detection requires \textit{a priori} knowledge of the evolution or optimization over initial conditions. There is no trajectory-level criterion that identifies the ``right'' pairs without additional global information.

Consequently, Non-Markovianity is a property of the global evolution rather than of individual paths. As shown in Example 3, even sampling many initial states together with their associated trajectories may fail to reveal Non-Markovian behavior. In such cases, full process tomography may be necessary. In fact, for any initial distribution of states, one can construct a family of Lindbladians, dependent on time and on the initial condition, that reproduces the observed trajectories of an arbitrary evolution. This construction follows naturally by considering the quantum–classical state
$\rho=\sum_j \mu_j\ket{j}\bra{j}\otimes\rho_j$
and applying the theorem on the enlarged space. This quantum–classical embedding is also natural in the context of noncommutative optimal transport. Moreover, replacing a single initial point by a small neighborhood around it and tracking the entire region along the trajectory still does not allow one to detect Markovianity.

The deeper origin of this indistinguishability lies in the geometry of the quantum state space (see supplementary material \cite{supp_material}). The state space does not admit a tangent vector space, but only a tangent cone. At boundary or corner points, certain Lindbladian directions, particularly those with strictly negative coefficients, are forbidden (see FIG \ref{figure}). 
\begin{figure}
    \centering
    \includegraphics[width=0.38\linewidth]{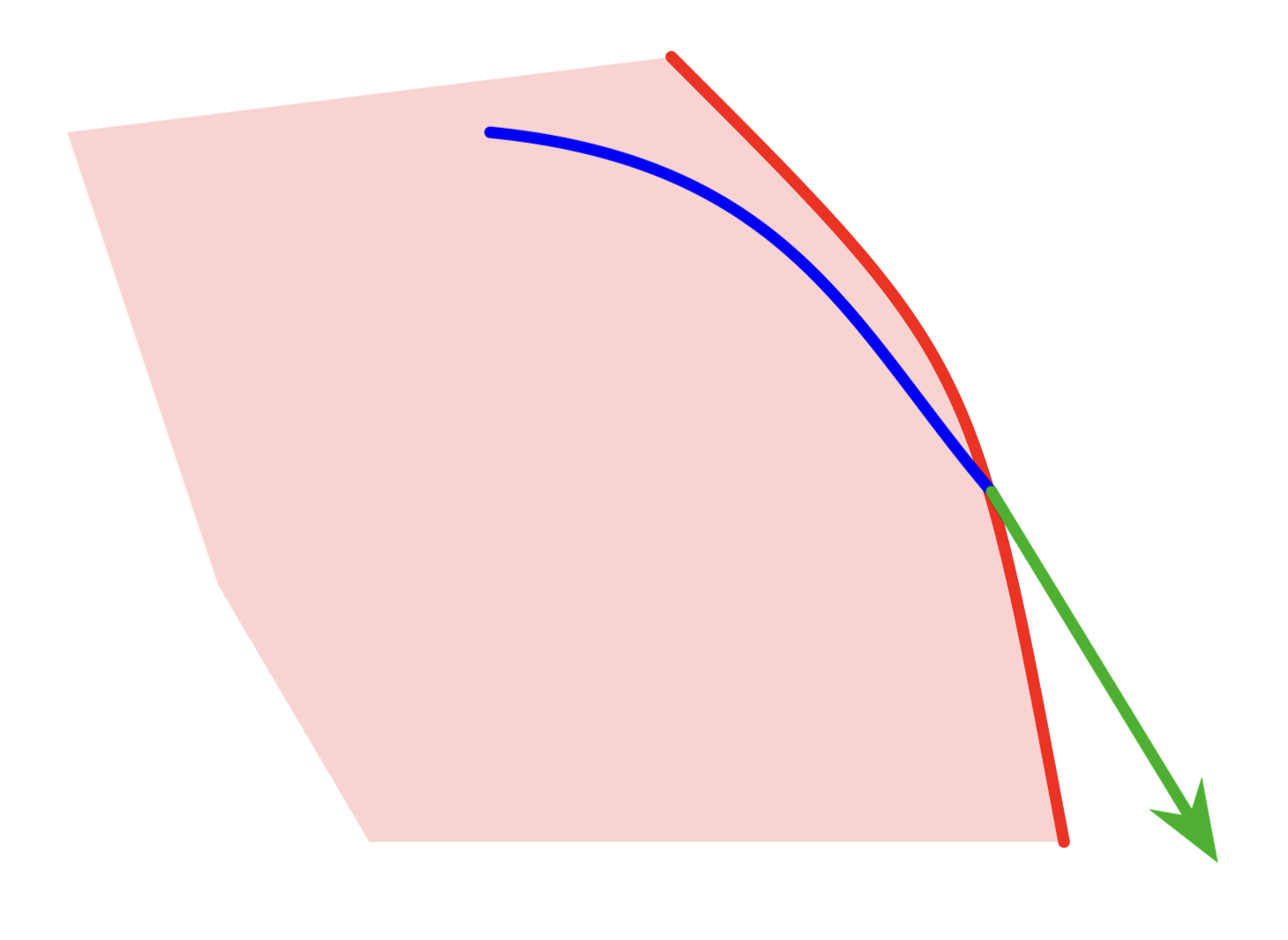}
    \includegraphics[width=0.35\linewidth]{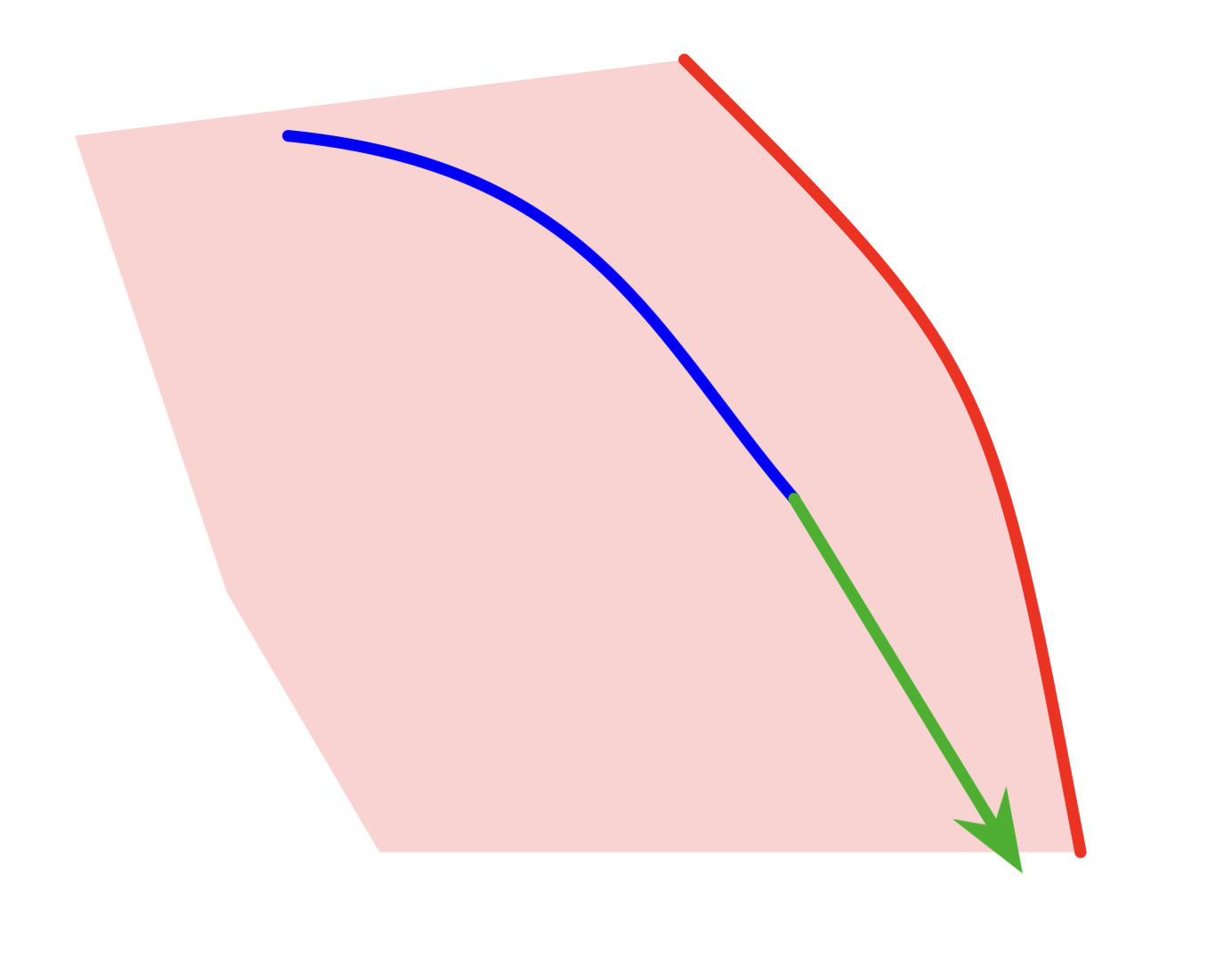}
    \caption{The blue line shows a path of densities with tangent given by the green arrow. On the left we see that a jump operator with negative sign can lead out of the space of densities if the path hits the boundary. On the right we see that the derivative stays inside the space of densities if path does not hit the boundary and can be lifted.}
    \label{figure}
    \vspace{-0.5cm}
\end{figure}

Markovianity may also fail at the endpoint of a closed time interval. Crucially, the construction of a time-local Lindbladian depends only on geometric properties of the trajectory, such as the stability of its eigenspaces. As a result, trajectories alone cannot distinguish between positive and completely positive evolutions, Markovian and Non-Markovian dynamics, or divisible and non-divisible processes.

Although Non-Markovianity does not generate additional trajectories, it remains a valuable resource for control tasks, including simulation, experimental design, and memory effects. Markovian evolutions that reproduce the same trajectories may require substantially larger resource sets, such as additional environment registers or more complex controls.

\paragraph*{Acknowledgments.}
    MJ is partially supported by NSF DMS 2247114.
    MJ also acknowledges the IPAM program on Non-commutative Optimal Transport for inspiration.

\vspace{-0.4cm}
\bibliography{nm_ref}

\clearpage
\onecolumngrid

\section*{Supplementary material for\\[0.5em]``The Hidden Nature of Non-Markovianity"}\label{supp}

\subsection{Tangent Cone and Lifting Problem}
We recall that the state space carries a natural stratified tangent geometry.
For $\rho\in\mc D(\mc H)$, we denote the tangent cone by
$$T_\rho^+\mc D(\mc H)=\lset \left.\dfrac{d}{dt}\right|_{t=0^+}\rho_t: \rho_0=\rho, \rho_t\in C^2([0,\infty), \mc D(\mc H))\rset.$$
The following characterization was obtained by the current authors \cite{control}:
\begin{equation}\label{tangent}
\begin{aligned}
    T_\rho^+\mc D(\mc{H}) = \lset L(\rho):L \text{ is Lindbladian}\rset.
\end{aligned}
\end{equation}
This is a characterization of the tangent cone as the set of kinematically admissible velocities at $\rho$, arising from all possible Lindblad generators evaluated at that point. It justifies why Lindbladian is the right notion to consider in this context. This is our motivation for \eqref{LL}.

\begin{question}[Lindbladian Lifting Problem] \label{question2}
    Given a path $\rho_t$ in $\mc D(\mc H)$, does there exist a family of Lindbladians $\{L_t\}_{t\ge0}$ such that $$\dot\rho_t = L_t(\rho_t)\, ?$$
    In other words, is the desired kinematic trajectory $\rho_t$ an integral curve of an admissible controlled Lindblad dynamics? A family of time-local Lindbladians satisfying this equation of motion is called a Lindbladian lift. 
\end{question}

\section{Existence of Reguarlized Lifts}
\subsection{Parametric and Geometric Lifts}
In classical geometric control and connection theory, horizontal lifting naturally separates into two notions: parametric lifting, which preserves the given time parametrization, and geometric lifting, which allows reparametrization. In other words, parametric liftings capture dynamical trajectories, whereas geometric liftings solely capture curves viewed as a set of points. This distinction arises because even very simple curves may fail to admit a lift in the stricter sense, while becoming liftable once their timing is adjusted.

Motivated by this classical picture, we begin with the more restrictive notion.
\begin{definition} We say that a path $\rho_t: I \to \mc D(H)$ has a \textbf{parametric lift} if there exists Lindbladians $L_t$ so that $\dot \rho_t = L_t(\rho_t)$ \end{definition}
\begin{example}
Take $\rho_t: [0,1) \to \mc D(H)$ with $\rho_t = \diag{(1-t,t)}$ with $\dot \rho_t = \diag{(-1,1)}$ Is is impossible to find a parametric lift for $\rho_t$. We write down an arbitrary Lindbladian that preserves the diagonal by $L_t = \sum \gamma_{ij}(t) L_{e_{ij}}$. We're going to assume that $\gamma_{ij}(t) \geq 0$, i.e. that the evolution is divisible. 
\begin{align*}
 \gamma_{00}(t) L_{e_{00}}\diag(1-t,t) &= \begin{pmatrix} 0 & 0 \\ 0 & 0\end{pmatrix} \\
\gamma_{01}(t)L_{e_{01}}\diag(1-t,t) &= \begin{pmatrix} \gamma_{01}(t)  t & 0 \\ 0 &  - \gamma_{01}(t) t \end{pmatrix} \\
\gamma_{10}(t)L_{e_{10}}\diag(1-t,t) &= \begin{pmatrix} - \gamma_{10}(t) (1-t)& 0 \\ 0 & \gamma_{10}(t)(1-t) \end{pmatrix} \\
\gamma_{11}(t)L_{e_{11}}\diag(1-t,t) &= \begin{pmatrix} 0 & 0 \\ 0 & 0\end{pmatrix}
\end{align*}

Hence, we obtain that 
$$L_t\diag(1-t,t) = \diag((\gamma_{01}(t) + \gamma_{10}(t))t -\gamma_{10}(t), (-\gamma_{01}(t)-\gamma_{10}(t))t + \gamma_{10}(t))$$

Write $\gamma_{01}(t) = \alpha$, $\gamma_{10}(t) = \beta$. So we have
\begin{equation*}
\begin{aligned}
    L_t\diag(1-t,t) &= \diag((\alpha + \beta)t -\beta, (-\alpha-\beta)t + \beta)\\
    &=  \diag(-1[(-\alpha-\beta)t + \beta], (-\alpha-\beta)t + \beta)\\
    &= \diag(-1,1)
\end{aligned}
\end{equation*}
To show this, we need that $-\alpha t - \beta t + \beta = 1$. First write $$\alpha t = - \beta t + \beta - 1$$ which gets us that $$\alpha = \frac{(1-t)\beta - 1}{t}.$$ From $\alpha > 0$ we have that $(1-t)\beta \geq 1$,  $$\beta \geq  \frac{1}{(1-t)}.$$ As $t \to 1$, $\beta$ must become unbounded and so there is no bounded parametric lift.
\end{example}

This computation shows that bounded parametric lifting fails solely because $\dot\rho_1$ lies outside $T_{\rho_1}^+ \mc D(\mc H)$. A single violation is enough to obstruct any bounded parametric lift. Still, the curve poses no intrinsic difficulty. Once we relax the parametrization, it admits a lift in the weaker geometric sense.

\begin{definition}
We say that a path $\rho_t: I \to \mc D(H)$ has a \textbf{geometric lift} if there exists a reparameterization $\rho_{f(t)}$ so that $\rho_{f(t)}$ has a parametric lift. 
\end{definition}
\begin{example}\label{moll}
$\rho_t = \diag(1-t,t)$ has a geometric lift. We take $f(t) = e^{1 - \frac{1}{(t-1)^2}}$ and write $\rho_{f(t)} = \diag(1-f(t),f(t))$ with $\dot \rho_t = f'(t) \diag(-1,1)$. We let $\mc R_{\diag(0,1)}$ be the replacer channel to the state $\diag(0,1)$ and let $$L_t = \frac{f'(t)}{f(t)-1} (\id - \mc R_{\diag(0,1)}).$$

We see that $$L_t(\rho_{f(t)}) = \frac{f'(t)}{f(t)-1}  [\diag(1-f(t),f(t)) - \diag(0,1)] = f'(t) \diag(-1,1)$$ as desired. Note that $\lim_{t \to 1} \frac{f'(t)}{f(t)-1} = 0$ so $L_t$ is bounded. Note that there are many different choices of $f(t)$ that can work here. 

\end{example}

It is clear that any path admitting a parametric lift also admits a geometric lift, simply by taking the identity reparametrization.

In what follows we consider both notions: we seek parametric lifts whenever they exist, and turn to geometric lifts when a parametric realization is obstructed.

\subsection{Measurable and Integrable Lifts}
 Integrability (and measurability) of the family $L_t$ are relevant for existence and uniqueness of the channels give by the equation of motion, $\dot T_t = L_t T_t$. 
\begin{proposition}\label{affine}
There exists a geometric lift for any piecewise affine path $\rho_t : I \to \mc D(H)$. 
\end{proposition}
\begin{proof}
Indeed, take $f(t)$ as in example \ref{moll}. Then for any affine path $\sigma_i \mapsto \sigma_{i+1}$ we can find an $L^i_t = \frac{f'(t)}{f(t)-1}(\id - \mc R_{\sigma_{i+1}}): [0,1) \to \mc D(H)$. If the path only has $n$ many pieces we write $L^k_t = 0$ for $k \geq n$. Let $\chi$ be the indicator function on $[0,1)$. We can then write $L_t = \sum_{i=0}^{\infty} \chi(t - i) L^i_t$. This achieves $f'(t) \dot \rho_{f(t)} = L_t(\rho_{f(t)})$. 
\end{proof}

\begin{proposition}
Let $\rho_t: I \to \mc D(H)$ be an arbitrary measurable path. For all $\varepsilon$, there exists $\eta_t: I \to \mc D(H)$ such that $\int_I \norm{\rho_t - \eta_t}_1 dt < \varepsilon$ where $\eta_t$ has a geometric lift.
\end{proposition}
\begin{proof}
Piecewise affine functions are dense in $L^1$. From example \ref{moll} we have that $\eta_t$ has a geometric lift.
\end{proof}

We have the following result from Cai, Govindarajan and Junge \cite{control} which allows for parametric lifts.
\begin{proposition}[spectral lifting] Let $\lambda(t)$ be the spectral gap of $\rho_t$. Assume $\rho_t, \dot \rho_t, {\dot \rho_t}^{-1} \in L^{\infty}(I,\mc D(\mc H))$ and $\dot \rho_t \in T^+_{\rho_t}(\mc D(\mc H))$ almost everywhere. If  $\lambda_{min}(t)^{-1}, \sqrt{\lambda_{min}(t)^{-1}} \in L^1(I, \mc D(\mc H))$ then $\rho_t$ has a measurable parametric lift.
\end{proposition}

\subsection{Continuous and Differentiable Lifts}
We can strengthen the spectral lifting result to allow for continuous lifts by smoothing over the points where the spectral gap goes to $0$. However, this smoothing operation will require the set of rank shifts of $\rho_t$ to be sparse. Let $\rankshift(\rho_t)$ be the set of discontinuities of $\rank(\rho_t)$. We have the following result:

\begin{theorem}\label{cont_para}
    Take $\rho_t \in C^2(I, \mc D(H))$ with the property that $\dot \rho_t \in T^+_{\rho_t}(\mc D(H))$. If $\rho_t$ has a continuous unitary diagonalization and $\rankshift(\rho_t)$ and has no accumulation points, then $\rho_t$ has a continuous parametric lift. 
\end{theorem}
\begin{proof}
We know that $\rho_t$ can be diagonalized by $U_t \rho_t U^*_t$ where $U_t$ is continuous. We will then assume, without loss of generality that $\rho_t$ is diagonal since we can always modify $\tilde L_t =  \frac{d}{dt}(Ad_{U_t} T_t)$ where $\dot T_t = L_t T_t$ is the semigroup that lifts $\rho_t$. For now we assume $\rho_t$ is of constant rank, we will weaken this assumption later. \\
We write $$\rho_t = \begin{pmatrix} \rho_{11}(t) & 0 & 0 \\ 0 & 0 & 0 \\ 0& 0 &0  \end{pmatrix} \qquad \dot \rho_t = \begin{pmatrix}x_{11}(t) & x_{12}(t) & 0 \\ x_{21}(t) & x_{22}(t) & 0 \\ 0& 0 &0  \end{pmatrix}$$
 On the support of $\rho$, the tangent vectors are giving by replacer channels $\frac{1}{\varepsilon_t}(\mc R_{\rho_t + \varepsilon_t x_{11}(t)} - \id)$. Outside the support of $\rho$ we can generate the tangent vectors by two jump operators: $$a_1(t) = \frac{1}{\sqrt{2}} x_{22}(t)\rho_{11}^{-1/2}(t) \qquad a_2(t) = -a_{1}(t)^{-1} \rho_{11}^{-1}(t) x_{12}(t) $$
If $\rho_t$ does not change rank such a path is continuous.

When $\rho_t$ changes rank we have to substitute the Lindbladian by requiring $\tilde L_t(\rho_t) = L_t(\rho_t)$ and $\tilde L_t$ continuous at $t \in \rankshift(\rho_t)$. To work up to the full result, let's first assume that there exists a single $t_0 \in \rankshift(\rho_t)$. Write $E_{\rho_t}$ for the projection onto the linear subspace of operators generated by $\rho_t$. Note that $E_{\rho_t}$ is continuous by assumption. We then define $$\tilde L_t(\rho_t) = L_t \circ E_{\rho_t} + L_{t_0} \circ (\id - E_{\rho_t}).$$ By assumption, $L_t$ is continuous when applied to $\rho_t$, and $L_{t_0}$ is continuous since it is constant. \\
Now we take the general case where $\rankshift(\rho_t)$ has no accumulation point. By Bolzano-Weierstrass, $\rho_t$ must be finite. For any two $t_i, t_{i+1}$ we pick $s,r$ so that $t_i < s < r < t_{i+1}$ and on this interval we take the interpolation $$\tilde L_t(\rho_t) = L_t \circ E_{\rho_t} + \lpara\frac{r-t}{r-s}L_{t_i} + \frac{t-s}{r-s} L_{t_{i+1}}\rpara \circ (\id - E_{\rho_t}).$$ On $[t_i, s)$ and $(r,t_{i+1}]$ we can then force $\tilde L_t(\rho_t)$ to be continuous. $\tilde L_t$ is then a continuous lift for $\rho_t$. 
\end{proof} 

\begin{corollary}
If $\rho_t$ has no non-zero repeated eigenvalues and is $C^k$, then it has a $C^k$ parametric lift.
\end{corollary}
\begin{proof}
 Because $\rho_t$ is positive and differentiable and has no non-zero repeated eigenvalues it can be diagonalized as $U_t \rho_t U^*_t$ with continuous unitaries $U_t$ \cite{lax2007, kato2013}. If $\rho_t$ is $C^k$ and has no-nonzero repeated eigenvalues it can be diagonalized by a $C^k$ unitary $U_t$. Since $\rho_{11}(t)$ is strictly positive, $a_1(t)$ and $a_2(t)$ are also $C^k$, and so the lift $L_t$ is $C^k$.
\end{proof}

\begin{corollary}
If $\rho_t$ is analytic, then it has an analytic parametric lift.
\end{corollary}
\begin{proof}
If $\rho_t$ is analytic then $U_t$ above is analytic, as are $a_1(t)$ and $a_2(t)$. The analyticity of $U_t$ prevents $\rankshift(\rho_t)$ from having accumulation points, so that the lift $\tilde L_t$ is analytic.
\end{proof}
In the full-rank case we do not have to worry about the shifting support of $\rho_t$. This eliminates the need for a continuous unitary diagonalization which greatly simplifies our necessary assumptions. 
\begin{corollary}\label{rank_lift}
If a $C^2(I, \mc D(H))$ path $\rho_t$ is of full rank for all $t \in I$ then it admits a continuous parametric lift.
\end{corollary}
\begin{proof}
For any $t>0$ there exists a $\varepsilon_t > 0$ with the property $\rho_t + \varepsilon_t \dot\rho_t > 0$. Moreover, this $\varepsilon_t$  can be chosen continuously. $L_t = \frac{1}{\varepsilon_t} (\mc R_{\rho_t + \varepsilon_t \dot \rho_t} - \id)$ is then a Lindbladian lift. 
\end{proof}

\end{document}